\documentclass[reqno]{amsart}
\usepackage{hyperref}
\usepackage{graphicx}
\usepackage{amssymb}
\usepackage{mathrsfs}

\begin{document}  \setcounter{page}{1}
\title[\hfilneg EJDE-2021/SI/01\hfil Nonlocal advection diffusion equation]
{The nonlocal advection diffusion equation and the two-slit experiment in quantum mechanics}

\author[Glenn Webb \hfil EJDE/Si/01 \hfilneg]
{Glenn Webb}

\address{Glenn Webb \newline
Mathematics Department,
Vanderbilt University,
Nashville, TN 37212 USA}
\email{glenn.f.webb@vanderbilt.edu}

\dedicatory{Dedicated to the memory of John W. Neuberger}


\begin{abstract}
A partial differential equation model is analyzed for the two-slit experiment of quantum mechanics.
The state variable of the equation is the probability density function of particle positions. The
equation has a diffusion term corresponding to the random movement of particles, and a
nonlocal advection term corresponding to the movement of particles in the transverse direction
perpendicular to their forward movement. The model is compared to the Schr\"odinger equation model of
the experiment. The model supports the ensemble interpretation of quantum mechanics.
\end{abstract}

\maketitle
\numberwithin{equation}{section}
\newtheorem{theorem}{Theorem}[section]
\newtheorem{lemma}[theorem]{Lemma}
\newtheorem{proposition}[theorem]{Proposition}
\newtheorem{corollary}[theorem]{Corollary}
\newtheorem{remark}[theorem]{Remark}
\newtheorem{definition}[theorem]{Definition}
\newtheorem{example}[theorem]{Example}
\allowdisplaybreaks

\newcommand{\iii}[1]{|\kern-0.25ex|\kern-0.25ex| #1 |\kern-0.25ex|\kern-0.25ex|}

\tableofcontents

\section{Introduction} \label{s:Intro}

The two-slit experiment demonstrates the fundamental probabilistic nature of quantum mechanics. In this experiment quantum particles are projected forward toward a screen with two parallel slits, and then observed on a detection screen further downstream (Figure \ref{Fig1}). 
An interference diffraction pattern of regularly spaced intensities is registered on the detection screen. The highest. density occurs in the center of the detection screen (Figure \ref{Fig2}), which is not the sum of the patterns observed for single slits separately (\cite{Tonomura}). The observed fringe pattern for two slits is characteristic of wave phenomena.

\begin{figure}[htp]
\begin{center}
\includegraphics[width=4.8in,height=1.5in]{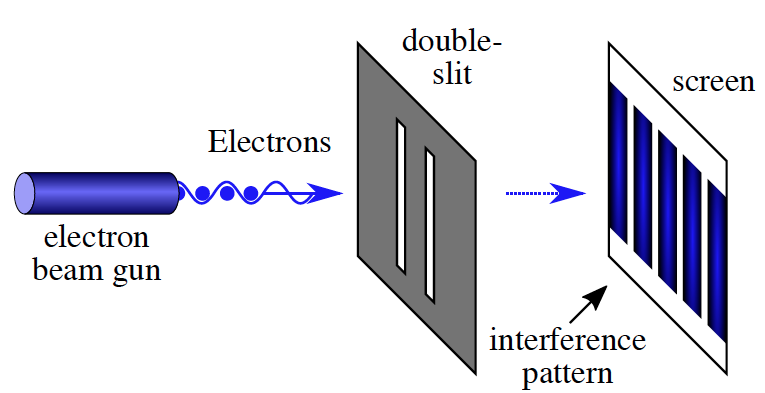}\\
\caption{The two slit experiment with electrons. wikipedia.org$/$wiki/Double-slit experiment}
\label{Fig1}
\end{center}
\end{figure}

\begin{figure}[htp]
\begin{center}
\includegraphics[width=4.6in,height=1.3in]{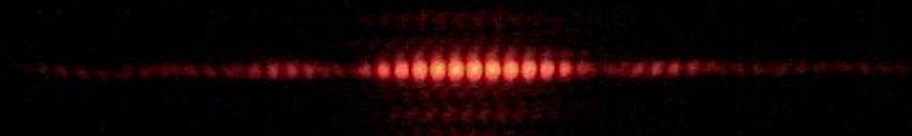}\\
\caption{The interference diffraction pattern of a two slit experiment with high order fringe patterns. wikipedia.org$/$wiki/Double-slit experiment}
\label{Fig2}
\end{center}
\end{figure}

In previous work we have investigated mathematical models for the two-slit experiment (\cite{Webb1}, \cite{Webb2}): the Schr\"odinger equation model (SE) and the nonlocal advection diffusion equation model (NLAD). In this work we extend these investigations to account for higher levels in the fringe patterns images in the experiment. We numerically simulate the SE and NLAD model outputs and compare the simulations to experimental data.

The organization of this paper is as follows:
In Section 2 we develop the SE model, in Section 3 we develop the NLAD model, in Section 4 we provide numerical simulations of both models and compare these outputs to experimental images, and in Section 5 we provide a summary of out work.

\section{The Schr\"odinger equation model }
\label{s:SE}

The one-dimensional time-dependent complex-valued 
Schr\"odinger equation is the foundational phenomenological model of quantum mechanics: 

\begin{equation}
\frac{\partial}{\partial t} \psi(x,t) \, = \, i \frac{\hbar}{2 m} \,  \frac{\partial^2}{\partial x^2}\psi (x,t), \, t>0, \, \, \psi(x,0) \, = \, \psi_0(x),  \, \, - \infty <  x <  \infty.  
\label{Eq2.1}
\end{equation}

\noindent
Here $\hbar$ is the reduced Planck's constant and $m$ is the particle mass, which without loss of generality can be assumed to satisfy $\frac{\hbar}{m}=1$. The interpretation of the solution is that $\int_{x_1}^{x_2} \rho(x,t) dx$ is the probability of finding a single particle in the interval $(x_1,x_2)$  at time $t$, where $\rho(x,t) =  | Re \, \psi(x,t) |^2 + | Im \, \psi(x,t) |^2$, and $\rho(x,0)$ 
is normalized so that $\int_{-\infty}^{\infty} \rho(x,0) dx = 1$. In this formulation, the interpretation of the state variable $\psi$ at time $t$ is sometimes applied to a single individual quantum particle. For this interpretation, the following questions arise:
What do the real and imaginary parts of $\psi$ represent for an individual particle? What is time $t$ in an experiment with randomly separated independent temporal events? What does the initial condition $\psi_0$ correspond to for single particles emitted one at a time? 

The general solution of  (\ref{Eq2.1}) (with scaling $\hbar / 2 m = 1$), is as in \cite{Goldstein}:

\begin{equation}
\psi(x,t)=\frac{1}{\sqrt{2 \pi i t}} \int_{-\infty}^{\infty} e^{{\frac{i (x-y)^2}{2t}}} \psi_0(y) dy,  \,\,\,
\psi_0 \in L^1((R;C),y^2 dy) \, \cap \, L^2(R;C),
\label{Eq2.2}
\end{equation}
where

\begin{equation}
\rho(x,t)=\frac{1}{2 \pi t} \Bigg( \bigg( \int_{-\infty}^{\infty} \cos \bigg( \frac{(x-y)^2}{2t} \bigg) Re \psi_0(y) - \sin \bigg( \frac{(x-y)^2}{2t} \bigg) Im \psi_0(y) dy \bigg)^2
\label{Eq2.3}
\end{equation}

\begin{equation}
+ \bigg( \int_{-\infty}^{\infty} \sin\bigg( \frac{(x-y)^2}{2t} \bigg) Re\psi_0(y) dy + \cos\bigg( \frac{(x-y)^2}{2t} \bigg) Im\psi_0(y) dy \bigg)^2 \Bigg), \\
\nonumber
\end{equation}
with 

\begin{equation}
\int_{-\infty}^{\infty} \rho(x,0) dx = \int_{-\infty}^{\infty} |\psi(x,0) |^2dx = 1,
\text{ which implies } \int_{-\infty}^{\infty} \rho(x,t) dx = 1, t \geq0.
\nonumber
\end{equation}

The probability amplitude $\rho(x,t)$  in  
(\ref{Eq2.3}) exhibits a two-phase pattern as $t$ advances. 
In the first phase, the initial information $\rho(x,0)$ evolves to an established pattern, in which the lower peaks in the fringe pattern lie almost on the $x$-axis.
In the second phase, this established pattern undergoes a space-time dilation as time advances. In the second phase the profile of $\rho(x,t)$  is propagated in the spatial $x$-direction at a constant speed. In \cite{Webb1} it is proved that for general initial data  $\psi_0$ in 
(\ref{Eq2.3}) the probability amplitude 
$\rho(x,t)$ in (\ref{Eq2.3}) satisfies the asymptotic space-time dilation property: uniformly for $x \in R, \, T>0, \, t \geq 1$
\begin{equation}
\Bigg| \rho(x,t  \,T) \,  - \,  \frac {1}{t} \rho(\frac{x}{t},T) \Bigg| \, \leq   \, \frac{\sqrt{2}}{\pi t T^2} 
\Bigg( \int_ {-\infty} ^{\infty} y^2 \, | \psi_0(y) | dy \Bigg) \Bigg( \int_ {-\infty} ^{\infty}  | \psi_0(y) | dy \Bigg),
\label{Eq2.4}
\end{equation}
which implies 
\begin{equation}
\lim_{t \rightarrow \infty}t \, \rho(x,t) = 
\lim_{T \rightarrow \infty}t T \, \rho(x,t T) = 
\frac{1}{2 \pi}  | \int_ {-\infty}^{\infty} \psi_0(y) dy |^2;
\label{Eq2.5}
\end{equation}
uniformly for $x \in R$.

\subsection{The Schr\"odinger Equation with Step Function Initial Data}  

In \cite{Webb2} an example  for  the Schr\"odinger equation applied to the two-slit experiment was given with initial data consisting of two rectangular step functions, symmetric about the origin of the $x$-axis. In this example, the centers of the two rectangular strips are taken as $s = \pm 1$, which can be applied generally by scaling the spatial variable $x$. The probability density $\rho(x,t)$ for this example evolves from the first phase to the second phase.
In Figure \ref{Fig3}, $\rho(x,t)$ is illustrated at $t =1/\pi$, which is 
approximately the value of $t$ at which the transition from the first phase to the second phase occurs.

\begin{figure}[htp]
\begin{center}
\includegraphics[width=4.8in,height=1.5in]{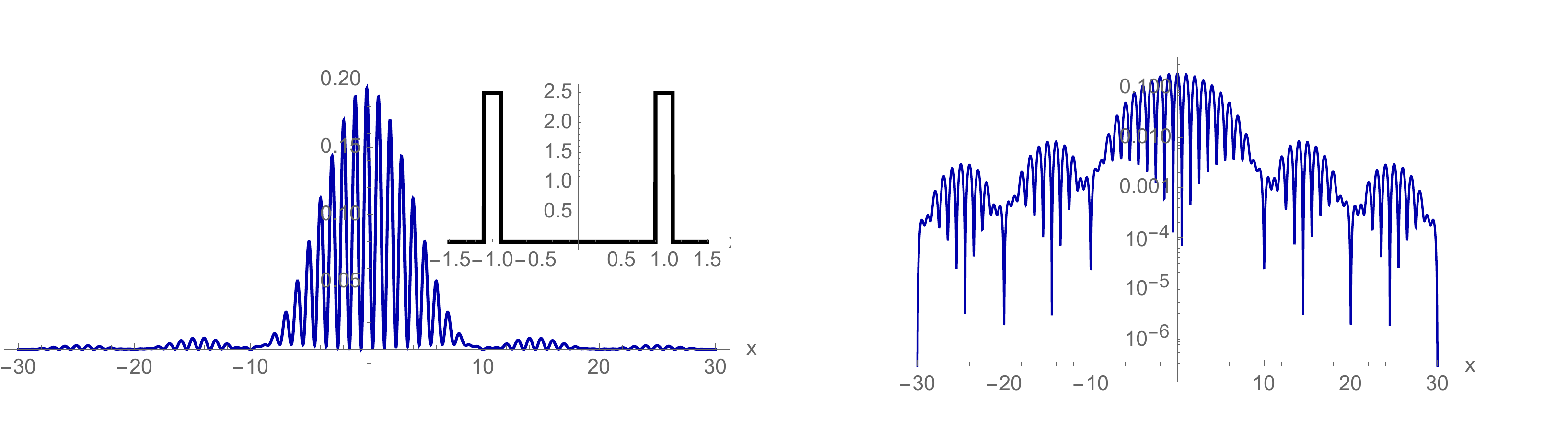}\\
\caption{$\rho(x,1/\pi)$ for $\rho(x,(0) =$ two rectangles with centers $s=\pm 1$, width $2 b, \, b = 0.1$, and height scaled so that 
$\int_{-s-b}^{-s+b}\rho(x,0) dx \,+ \, \int_{s-b}^{s+b}\rho(x,0) dx = 1$.
Left: $\rho(x,1 / \pi)$, the inset is the graph of $\rho(x,0)$.
The coordinates $x = \pm 10, \pm 20, \pm 30, \dots$
are local maxima for $\rho(x,1 / \pi)$ and $x \approx \pm 15, \pm 25, \pm 35, \dots $
are  local minima for $\rho(x,1 / \pi)$.
Right: $\log \rho(x,1 / \pi)$.}
\label{Fig3}
\end{center}
\end{figure}

\begin{remark}
In Figure \ref{Fig3} it is seen that the fringe pattern  has multiple interference pattern groupings on either side of a central principal  pattern. 
The diffraction pattern of $\rho(x,t)$ has infinitely many regularly spaced intervals
$$
\dots, [- 30, - 20], \, [- 20, - 10], [- 10, 10], [10 ,20],  [20, 30], \dots,
$$
with approximately regularly spaced interior fringes repeated in each of the intervals.
This illustration is consistent with the experimental data in Figure \ref{Fig2}.
The spacing is determined by $s/b = 10$.
\end{remark}

\section{The nonlocal advection diffusion equation model}
\label{s:NLAD}

An alternative to  the Schr\"odinger equation formulation of the two slit experiment is the nonlocal advection diffusion equation formulation NLAD. This formulation incorporates the movement at a given spatial $x$ coordinate as dependent on nearby spatial $x$ coordinates. Such models have been developed in \cite{Aharonov} and \cite{Webb2}. The NLAD model supports the ensemble interpretation of quantum mechanics, which maintains that mathematical description  of particle behavior  should correspond to communal particle behavior, rather than to individual particle behavior \cite{Einstein}.

In \cite{Aharonov} the authors examine an interpretation of the two-slit experiment based on the nonlocal interaction of a single particle with both slit openings. In their interpretation a single particle only passes through one slit, but with its corpuscular motion affected by the other slit. The authors provide a deterministic nonlocal dynamic equation of motion for the density of particles and relate the behavior of the solutions to the interference patterns observed in the two-slit experiment. This interpretation provides an alternative to the classical Schr\"odinger equation interpretation of this experiment as an ensemble wave of particles passing through both slits.

In the NLAD model in \cite{Webb2}, the nonlocal advection term represents  directed particle movement 
due to the influence of nearby particles, and the diffusion term represents variability of particle movement due to stochastic variation. 
The nonlocal advection term involves an integral corresponding to the slit separation widths $s$ in the $\pm x$-directions. 
In \cite{Webb2}, the NLAD equation analyzed was the following: 

\begin{align}
\frac{\partial}{\partial t} \omega(x,t) \, = 
\, \alpha \,  \frac{\partial^2}{\partial x^2} \omega(x,t) 
+   \frac{\partial}{\partial x} \int_{-s}^{s}
\, \beta_0 \, \omega(x+\hat{x},t) \frac{\hat{x}}{|\hat{x}|}d\hat{x} 
\hspace{3.7in}
\label{Eq3.1}
\\ 
=   \alpha \, \frac{\partial^2}{\partial x^2} \omega(x,t)  
+  \beta_0 \bigg( \omega(x + s,t)  - 2  \omega(x,t) + \omega(x - s,t) \bigg), 
\, t > 0, \, - \infty \, < x < \infty,
\hspace{2.5in}
\label{Eq3.2}
\end{align}
$$
\omega(x,0) \, = \, \omega_0(x), - \infty \, < x < \infty, \, \omega_0 \in L_+^1(-\infty,\infty), \, \int_{-\infty}^{\infty}\omega_0(x) dx =1. 
$$

In equation (\ref{Eq3.1}) $x$ is the spatial coordinate of particles, $t$ is the downstream distance perpendicular to the slits openings, 
$\alpha$ is the diffusion parameter, $\beta_0$ is the advection parameter, $2 s$ is the slit width, $\omega_0(x)$ is the initial data, and 
$\omega(x,t)$ is the probability density function for the distribution of particle positions. 
The values of $\alpha$ and $\beta_0$ are chosen so that the solutions of (\ref{Eq2.1}) and (\ref{Eq3.1}) are similar for $x \in [-20,20]$ at $t=1/\pi$.

In this work we analyze an extension of equation (\ref{Eq3.2}) to include higher order fringe pattern levels as seen in Figure (\ref{Fig2}). The extension of (\ref{Eq3.2}) to include higher order fringe pattern levels was illustrated in (\cite{Webb3}) with numerical examples.
Equation (\ref{Eq3.2}) is modified as follows:

\begin{align}
\frac{\partial}{\partial t} \omega(x,t) \, = 
\, \alpha \,  \frac{\partial^2}{\partial x^2} \omega(x,t) 
+   \frac{\partial}{\partial x} \int_{-s}^{s}
\, \beta_0 \, \omega(x+\hat{x},t) \frac{\hat{x}}{|\hat{x}|}d\hat{x} 
\hspace{4.4in} \\
 +   \frac{\partial}{\partial x} \int_{-\frac{3}{2b}}^{\frac{3}{2b}}
\beta_1 \, \omega(x+\hat{x},t) \frac{\hat{x}}{|\hat{x}|}d\hat{x} \,
+   \frac{\partial}{\partial x} \int_{-\frac{5}{2b}}^{\frac{5}{2b}}  \beta_2 \, \omega(x+\hat{x},t) \frac{\hat{x}}{|\hat{x}|}d\hat{x} 
\hspace{3.7in}
\nonumber
\\ 
=   \alpha \, \frac{\partial^2}{\partial x^2} \omega(x,t)  
+  \beta_0 \bigg( \omega(x + s,t)  - 2  \omega(x,t) + \omega(x - s,t) \bigg)   \hspace{3.7in}
\nonumber 
\\
+  \beta_1 \bigg( \omega(x + \frac{3 s}{2b},t)  - 2  \omega(x,t) + \omega(x - \frac{3 s}{2b},t) \bigg)
\hspace{5.0in}
\nonumber
\\
+  \beta_2 \bigg( \omega(x + \frac{5 s}{2b},t)  - 2  \omega(x,t) + \omega(x - \frac{5 s}{2b},t) \bigg),
\, t > 0, \, - \infty \, < x < \infty,
\label{Eq3.4}
\hspace{3.4in}
\\
\nonumber
\omega(x,0) \, = \, \omega_0(x), - \infty \, < x < \infty, \, \omega_0 \in L_+^1(-\infty,\infty), \, \int_{-\infty}^{\infty}\omega_0(x) dx =1.  \hspace{3.5in}
\end{align}
Equation (\ref{Eq3.4}) has two additional fringe pattern levels on either sided of the origin: 
$[- \frac{4 s}{b},- \frac{3 s}{b}]$,  
$[- \frac{3 s}{b},- \frac{2 s}{b}]$,
$[\frac{2 s}{b}, \frac{3 s}{b}]$,
$[\frac{3 s}{b}, \frac{4 s}{b}]$.
Additional levels can be added.

\subsection{Analysis of Equation (3.4)}

Let $X \, = \, L^1(-\infty,\infty)$, the space of integrable functions on $(-\infty,\infty)$ with norm
$\| f \| \, = \, \int_{-\infty}^{\infty}|f(x)| dx$. For $\sigma > 0$ let 

\begin{align}
&(T_{\sigma}(t)f)(x) = 
\label{Eq3.5}
\nonumber
\\
&\frac{1}{2 \, \sqrt{\sigma \, t}} \, \int_{-\infty}^{\infty}  \exp\bigg(- \frac{(x  -  y)^2}{4 \, \sigma \, t}\bigg) \, f(y) \, dy, \, \, f \in X, \, \, t \,  > \, 0, \, \, -\infty \, < x \, < \, \infty.
\end{align}
$T_{\sigma}(t), t \geq 0$ is a strongly continuous holomorphic semigroup of positive linear operators in $X$ with infinitesimal generator $(A_{\sigma} f)(x) = \sigma \, d^2 f(x) / d x^2$ satisfying $| T_{\sigma}(t) | \leq 1, t \geq 0$  (\cite{Kato}). 
Further, $(T_{\sigma}(t)f)(x)$ is the strong solution in $X$ to the diffusion equation 
\begin{equation}
\label{Eq3.6}
\frac{\partial}{\partial x}u(x,t) \, = \, \sigma \, \frac{\partial^2}{\partial x^2} u(x,t), \, \, u(x,0) \, = \, f(x), \, \, t \,  > \, 0, \, \, -\infty \, < x \, < \, \infty, \, f \in X.
\end{equation}
and (\cite{Yosida})
\begin{equation}
\label{Eq3.7}
\int_{-\infty}^{\infty} (T_{\sigma}(t) f) (x) dx \, = \,    \int_{-\infty}^{\infty} f(x) dx f \in X.
\end{equation}
For a bounded linear operator $B$ in $X$ define the exponential of $B$ as
$$ \exp(t B) f  \, = \, \sum_{n=0}^{\infty} \frac{(t B)^n}{n!} f, \, f \in X, \, t \geq 0. $$
Define the bounded linear operators $B_{0,\pm},  B_{1,\pm},  B_{2,\pm}$ in $X$ as follows:
for $f \in X$,  $-\infty  < x  <  \infty$,
$$(B_{0,\pm}f)(x)  = \beta_0 f(x \pm s),$$ 
$$(B_{1,\pm}f)(x)  = \beta_1 f(x \pm \frac{3 s}{2 b}),$$
$$(B_{2,\pm}f)(x)  = \beta_2 f(x \pm \frac{5 s}{2 b}).$$
Define the bounded linear operators $B_0, \, B_1, \, B_2$ in $X$ as follows:
$$(B_0 f)(x) = \beta_0 \bigg( f(x+s) - 2 f(x) + f(x-s) \bigg),$$
$$(B_1 f)(x) = \beta_1 \bigg( f(x+\frac{3 s}{2 b}) - 2 f(x) + f(x-\frac{3 s}{2 b})) \bigg),$$
$$(B_2 f)(x) = \beta_2 \bigg( f(x+s) - 2 f(x) + f(x-s) \bigg).$$
Then, since $B_{0,+}, \, B_{0,-}$, $B_{1,+}, \, B_{1,-}$, $B_{2,+}, \, B_{2,-}$ commute,
$$exp(t B_0) = e^{- 2 \beta_0 t} exp(t B_0) =  e^{- 2 \beta_0 t} exp(t B_{0,+}) exp(t B_{0,-}),$$
$$exp(t B_1) = e^{- 2 \beta_1 t} exp(t B_1) =  e^{- 2 \beta_1 t} exp(t B_{1,+}) exp(t B_{1,-}),$$
$$exp(t B_2) = e^{- 2 \beta_0 t} exp(t B_2) =  e^{- 2 \beta_2 t} exp(t B_{2,+}) exp(t B_{2,-}).$$

\begin{theorem}
Let $X \, = \, L^1(-\infty,\infty)$, let $\alpha, \, s, \, b, \, \beta_0, \beta_1, \beta_2 > 0$,
Let $T_{\alpha}(t),t \geq 0$ be the semigroup of linear operators as in equation (\ref{Eq3.5}).
The unique generalized solution of equation (\ref{Eq3.4}) is given by the strongly continuous semigroup of linear
operators $T(t), t \geq 0$ in $X$ with
\begin{equation}
\label{Eq3.8}
T(t) \omega_0 \, = \, T_{\alpha}(t) \, exp(t B_0) \, exp(t B_1) \, exp(t B_2) \, \omega_0, \,  t \geq 0, \, \omega_0 \in X.
\end{equation}
Further, if 
\begin{equation}
\label{Eq3.9}
\omega_0(x) \geq 0 \, \text{ a.e. on} \,  (-\infty,\infty) \, \text{and} \, \int_{-\infty}^{\infty} \omega_0(x) dx \,= \, 1,
\end{equation}
then
\begin{equation}
(T(t) \omega_0) (x) \geq 0 \, \text{ a.e. on} \,  (-\infty,\infty) \, \text{and} \, \int_{-\infty}^{\infty} (T(t) \omega_0)(x) dx \,= \, 1.
\label{Eq3.10}
\end{equation}
\end{theorem}

\begin{proof}
For $\omega_0 \in X$, $\omega_0(x) \geq 0 \, \text{ a.e. on} \,  (-\infty,\infty)$,
$$(B_{0,+} \omega_0)(x)  = \beta_0 \,  \omega_0(x + s)  \,\geq 0 \, \text{ a.e. on} \,  (-\infty,\infty),$$
$$(B^2_{0,+} \omega_0)(x)  = B_{0,+} \, ( \beta_0 ( \omega_0(x + s) ) = \beta_0^2  \, \omega_0(x + 2 s) \,  \geq 0 \, \text{ a.e. on} \,  (-\infty,\infty),$$
$$(B^3_{0,+} \omega_0)(x)  = B^2_{0,+} \, ( \beta_0 ( \omega_0(x + s) )
= \beta_0^3  \, \omega_0(x + 3 s) \,  \geq 0 \, \text{ a.e. on} \,  (-\infty,\infty) \dots.$$
Thus, 
$$ \exp(t B_{0,+}) \omega_0  \, = \, \sum_{n=0}^{\infty} \frac{t^n \,  B_{0,+}^n}{n!} \omega_0  \geq 0 \, \text{ a.e. on} \,  (-\infty,\infty).$$
Similarly, $\exp(t B_{0,-}) \omega_0$, 
$\exp(t B_{1,+}) \omega_0$, 
$\exp(t B_{1,-}) \omega_0$, 
$\exp(t B_{2,+}) \omega_0$, 
$\exp(t B_{2,-}) \omega_0$ 
$\geq 0 \, \text{ a.e. on} \,  (-\infty,\infty)$.
From (\cite{Kato}) $T_{\alpha}(t) \omega_0 (x) \geq 0 \, \text{ a.e. on} \, (-\infty,\infty)$.
Thus, for $t \geq 0$,
$$T(t) \omega_0 \, = \, T_{\alpha}(t) \, exp(t B_0) \, exp(t B_1) \, exp(t B_2)  \geq 0$$. 

Let $\omega_0 \in X$, $\omega_0(x) \geq 0 \, \text{ a.e. on} \,  (-\infty,\infty)$, and $\int_{-\infty}^{\infty} \omega_0(x) dx \,= \, 1$.
Then,
$$\int_{- \infty}^{\infty} (B_{0} \omega_0) (x) dx
 \, = \, \beta_0 \, \Bigg( \int_{- \infty}^{\infty} \Bigg( \omega_0 (x + s) - 2 \omega_0 (x ) +  \omega_0 (x - s) \bigg)dx \, = \, 0,$$
$$\int_{- \infty}^{\infty} (B^2_{0} \omega_0) (x) dx
 \, = \, \beta_0 \, \Bigg( \int_{- \infty}^{\infty} \bigg( \omega_0 (x + 2 s) - 2 \omega_0 (x+ s) +  \omega_0 (x ) \bigg) \hspace{1.5in}$$
$$ - 2 \bigg( \omega_0 (x+ s) -2 \omega_0(x) \, + \omega_0(x -s) \bigg)
 + \bigg( \omega_0 (x -2  s) - 2 \omega_0 (x- s) +  \omega_0 (x) \bigg) \Bigg) dx \, = \, 0,$$
$$ \int_{- \infty}^{\infty} (B^3_{0} \omega_0) (x) dx = 0, \,  \dots. \hspace{3.5in}$$
Similarly, 
$$ \int_{- \infty}^{\infty} (B^n_{1} \omega_0) (x) dx = 0, \, \int_{- \infty}^{\infty} (B^n_{2} \omega_0) (x) dx = 0, \, n =1,2,3 \dots.$$
Then, 
$$\int_{- \infty}^{\infty} (exp(t B_0) \omega_0) (x) dx 
\, = \, \sum_{n=0}^{\infty} \frac{t^n}{n!}  \int_{- \infty}^{\infty} (B_0^n \omega_0) (x) dx
\, = \,  \int_{- \infty}^{\infty}  \omega_0 (x) dx \, = \, 1.$$
Similarly, 
$$\int_{- \infty}^{\infty} (exp(t B_1) \omega_0) (x) dx \, = \, 1, \, \text{ and } \,  \int_{- \infty}^{\infty} (exp(t B_2) \omega_0) (x) dx \, = \, 1.$$
From (\ref{Eq3.7})
$$\int_{- \infty}^{\infty} (T_{\alpha}(t) \omega_0) (x) dx \, = \, 1, \, t \geq 0.$$ 
Thus, for $t \geq 0$,
$$ \int_{- \infty}^{\infty} T(t) \omega_0 dx\, = \, \int_{- \infty}^{\infty} T_{\alpha}(t) \, exp(t B_0) \, exp(t B_1) \, exp(t B_2) dx \, = \, 1$$. 
\end{proof}

\section{Numerical simulations of the models}
\label{s:numerical}

In this section we provide numerical simulations of the solutions $\rho(x,t)$ of the SE model (\ref{Eq2.3})  
and the solutions $\omega(x,t)$ of the NLAD model (\ref{Eq3.4}). 
The SE solutions and the NLAD solutions are similar, depending on the values of parameters and the values of $t$.
There are two significant differences in the model outputs, which we will demonstrate in the simulations.

One difference is in the spacing of the local minima and maxima in the fringe patterns in the first phase. This spacing is very regular in the solutions of the NLAD model, but irregular in the solutions of the SE model. Another difference is in the second phase, where the solutions of the SE model demonstrate the space-dilation property, but the solutions of the NLAD model elevate above the $x$-axis and demonstrate a dispersion property of the fringe pattern peaks.

In Figures (\ref{Fig4}), (\ref{Fig5}), (\ref{Fig6}), (\ref{Fig7}),  simulations  are given for the first phase of the SE and NLAD models at four different time values.
The spacing of peaks for the NLAD model is very regular, whereas the spacing of peaks for the SE model is very irregular. In Figure (\ref{Fig7}) the spacing of the local minima of $\omega(x,1 / \pi)$ occurs regularly at $x \approx .5, 1.5, 2.5, 3.5, \dots$. The spacing of local minima of $\rho(x,1 / \pi)$ occurs irregularly  at $x \approx .5, 1.5,  2.5, 3.5, \dots, 7.5, 8.5, 9.25, 10.0, 10.75, 11.5, 12.5, \dots$, $17.5, 18.5, 19.25, 20.0$, $20.75, 21.5, \dots$.

\begin{figure}[htp]
\begin{center}
\includegraphics[width=4.6in,height=4.3in]{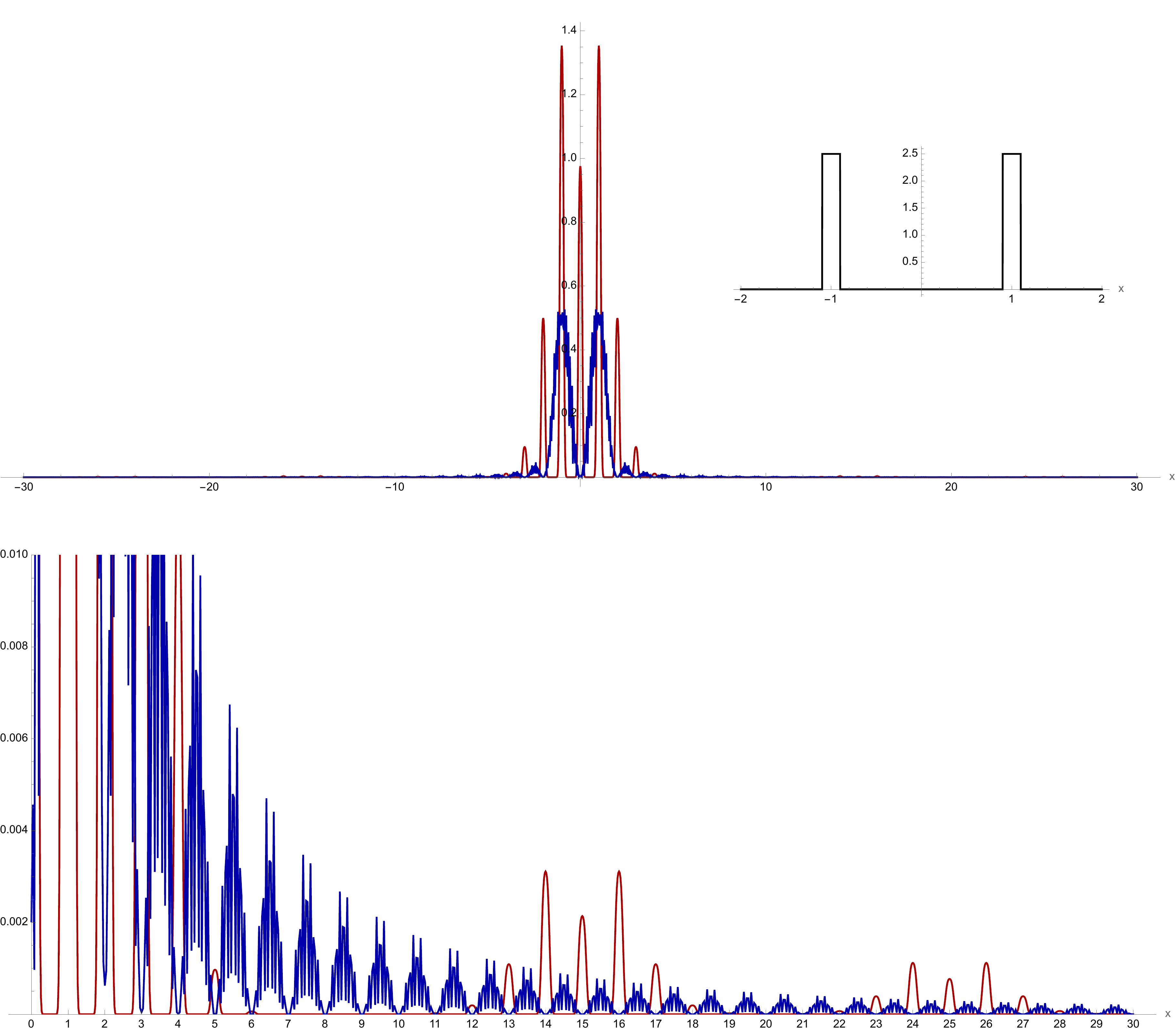}\\
\caption{The interference diffraction pattern of the solution $\rho(x,t)$ of SE  (blue) in equation (\ref{Eq2.3}) and the solution $\omega(x,t)$ of NLAD
(red) in equation (\ref{Eq3.4}).
$\alpha = 1/\pi^3$, $s = 1$, $b =.1$, $\beta_0 = 1 / (8 \, b^2)$, $\beta_1 = \pi / (2 \, b \, 15^2)$, $\beta_2 =  \pi / (2 \, b \, 25^2)$, $t = .1 / \pi$.
The bottom graphs are $log(\rho(x,t))$  (blue) and $log(\omega(x,t))$ (red).}
\label{Fig4}
\end{center}
\end{figure}

\begin{figure}[htp]
\begin{center}
\includegraphics[width=4.6in,height=4.3in]{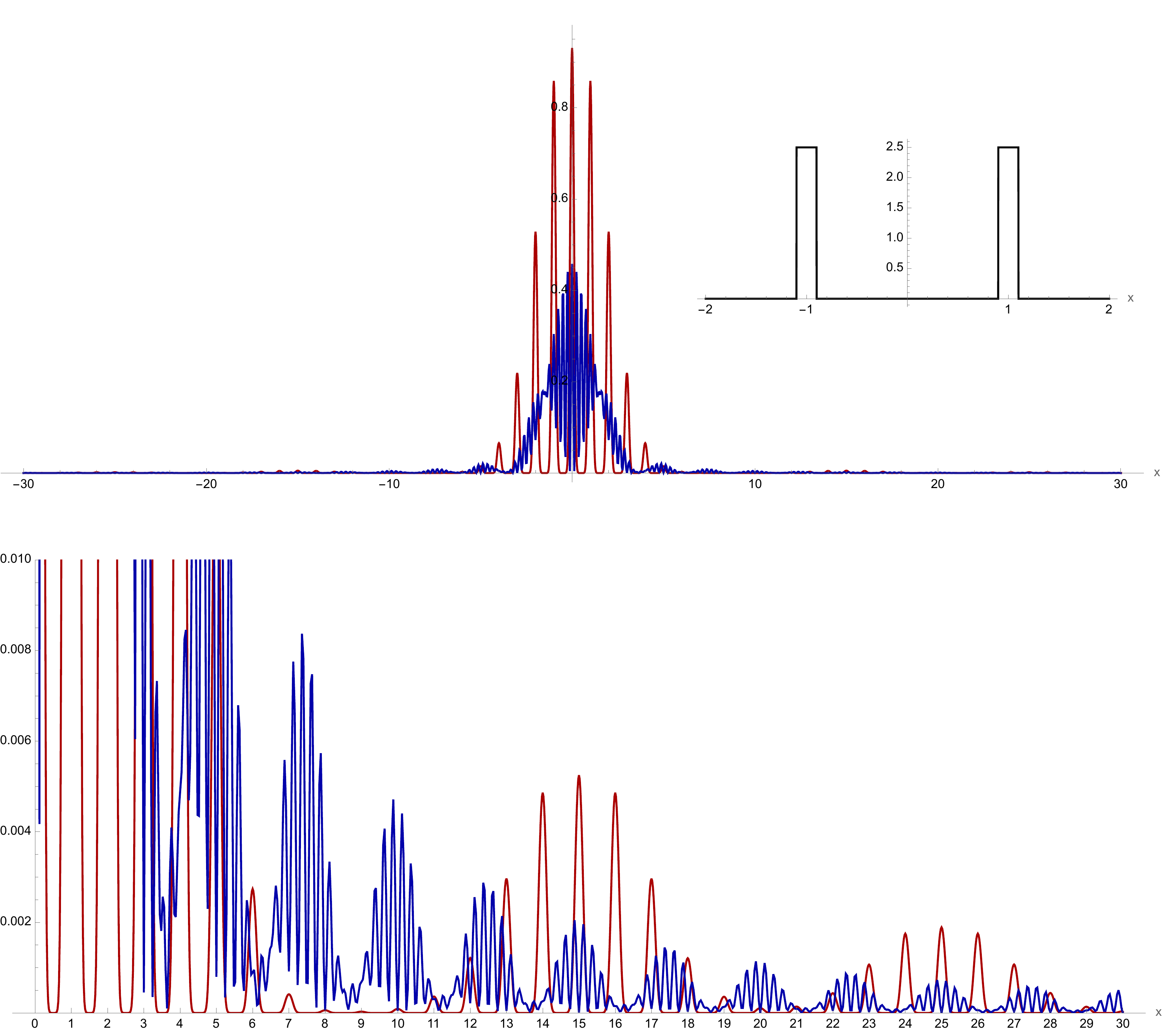}\\
\caption{The interference diffraction pattern of the solution $\rho(x,t)$ of SE  (blue) in equation (\ref{Eq2.3}) and the solution $\omega(x,t)$ of NLAD
(red) in equation (\ref{Eq3.4}).
$\alpha = 1/\pi^3$, $s = 1$, $b =.1$, $\beta_0 = 1 / (8 \, b^2)$, $\beta_1 = \pi / (2 \, b \, 15^2)$, $\beta_2 =  \pi / (2 \, b \, 25^2)$, $t = .25 / \pi$.
The bottom graphs are $log(\rho(x,t))$  (blue) and $log(\omega(x,t))$ (red).}
\label{Fig5}
\end{center}
\end{figure}

\begin{figure}[htp]
\begin{center}
\includegraphics[width=4.6in,height=4.3in]{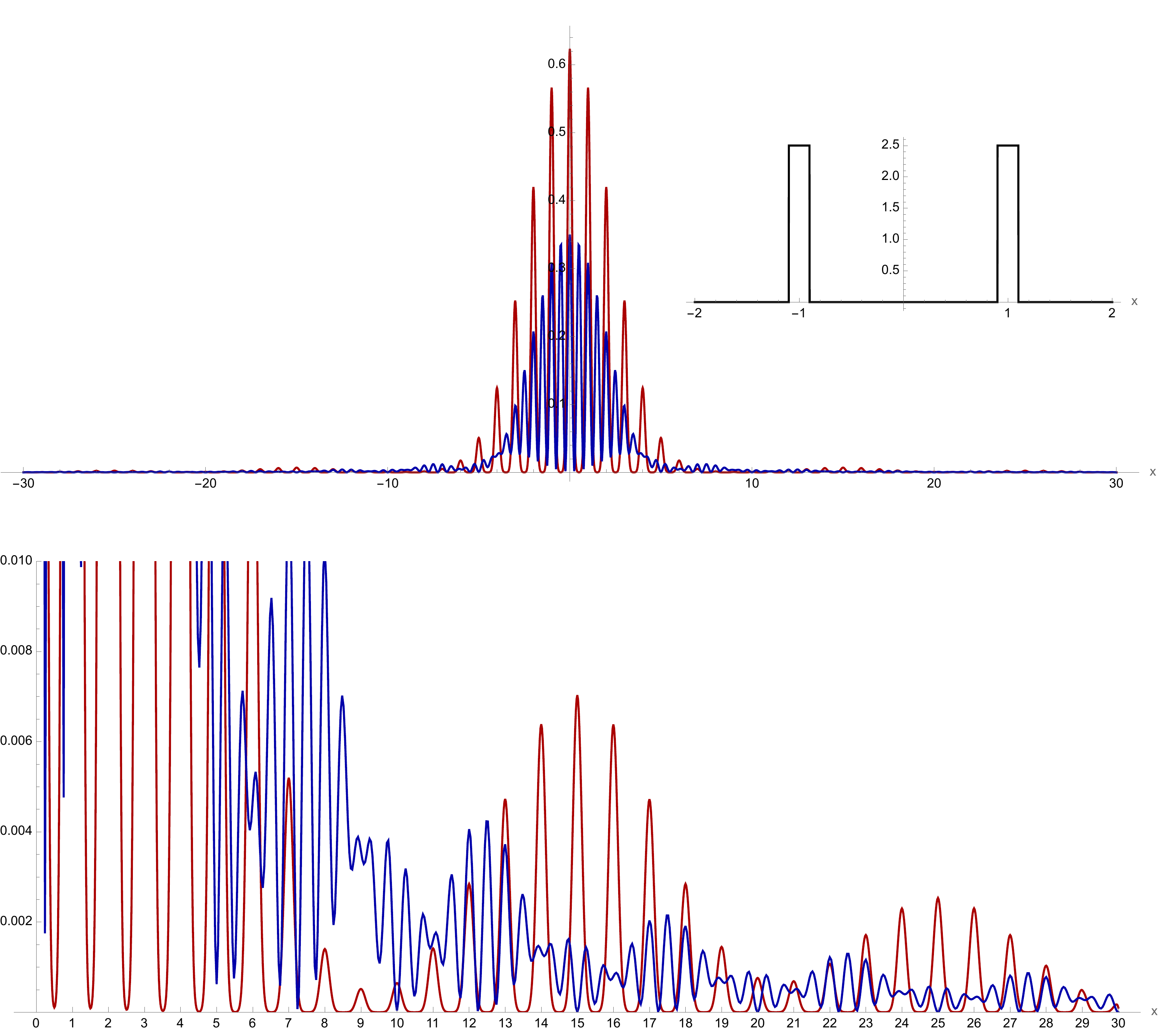}\\
\caption{The interference diffraction pattern of the solution $\rho(x,t)$ of SE  (blue) in equation (\ref{Eq2.3}) and the solution $\omega(x,t)$ of NLAD
(red) in equation (\ref{Eq3.4}).
$\alpha = 1/\pi^3$, $s = 1$, $b =.1$, $\beta_0 = 1 / (8 \, b^2)$, $\beta_1 = \pi / (2 \, b \, 15^2)$, $\beta_2 =  \pi / (2 \, b \, 25^2)$, $t = .5 / \pi$.
The bottom graphs are $log(\rho(x,t))$  (blue) and $log(\omega(x,t))$ (red).}
\label{Fig6}
\end{center}
\end{figure}

\begin{figure}[htp]
\begin{center}
\includegraphics[width=4.6in,height=4.3in]{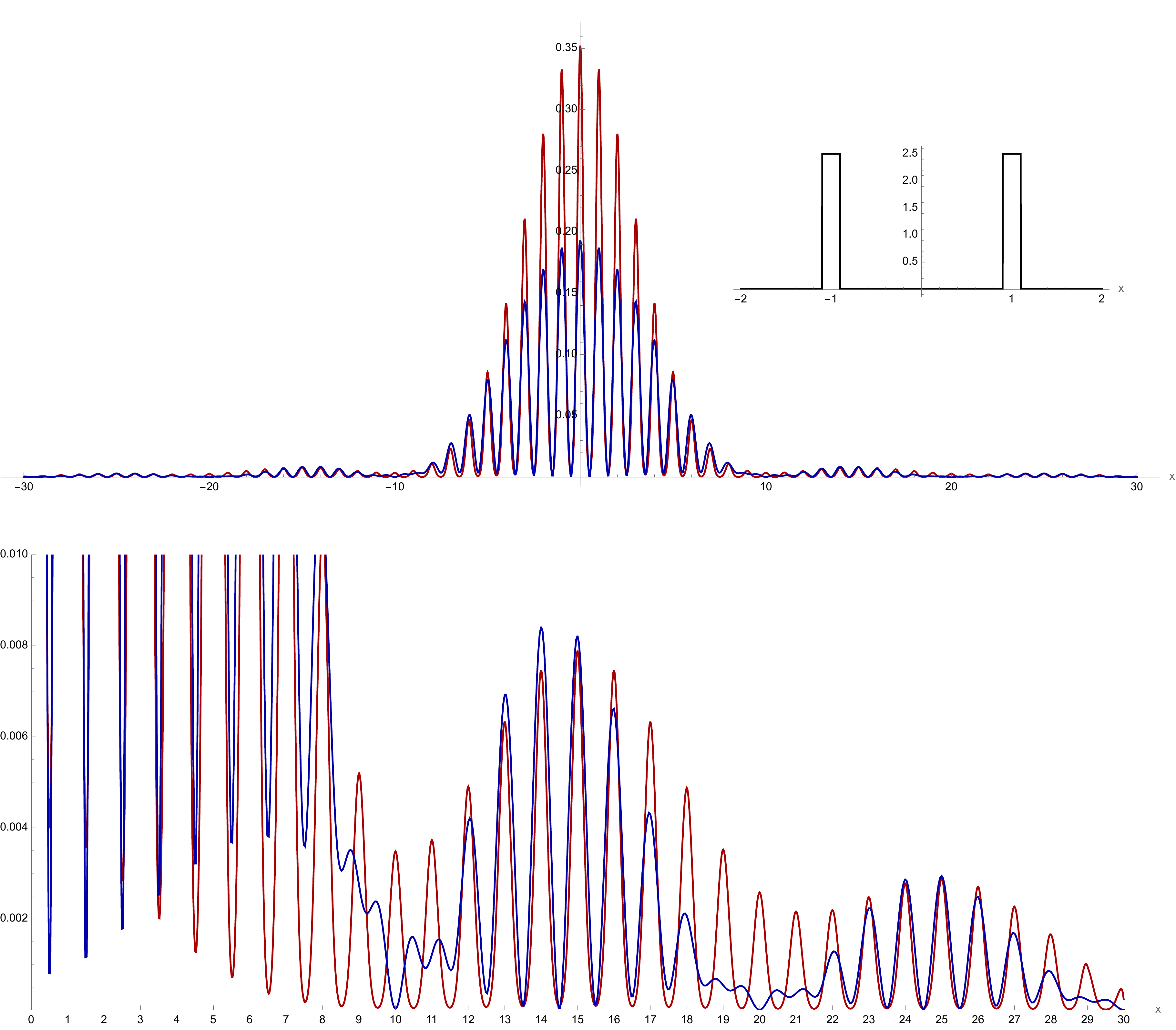}\\
\caption{The interference diffraction pattern of the solution $\rho(x,t)$ of SE  (blue) in equation (\ref{Eq2.3}) and the solution $\omega(x,t)$ of NLAD
(red) in equation (\ref{Eq3.4}).
$\alpha = 1/\pi^3$, $s = 1$, $b =.1$, $\beta_0 = 1 / (8 \, b^2)$, $\beta_1 = \pi / (2 \, b \, 15^2)$, $\beta_2 =  \pi / (2 \, b \, 25^2)$, $t = 1.0 / \pi$.
The bottom graphs are $log(\rho(x,t))$  (blue) and $log(\omega(x,t))$ (red).}
\label{Fig7}
\end{center}
\end{figure}

In Figures (\ref{Fig8}), (\ref{Fig9}), (\ref{Fig10}), (\ref{Fig11}), 
simulations  are given for the second phase of the SE and NLAD models at four different time values $t_2 = 2 / \pi, \, t_3 = 3 / \pi,  \, t_4 = 4 / \pi, \, t_6 = 6 / \pi$. In these simulations for the second phase, the solutions $\omega(x,t)$ of the nonlocal advection-diffusion equation (\ref{Eq3.4}) are space-time dilated according to the formulas
$
\hat{\omega}(x,t_2) =\frac{1}{2} \, \omega(\frac{x}{2}, t_2), \, \, \, \,
\hat{\omega}(x,t_3) =\frac{1}{3} \, \omega(\frac{x}{3}, t_3), \, \, \, \,
\hat{\omega}(x,t_4) =\frac{1}{4} \, \omega(\frac{x}{4}, t_4), \, \, \, \,
\hat{\omega}(x,t_6) =\frac{1}{6} \, \omega(\frac{x}{6}, t_6).
$
This dilation of $\omega(x,t)$ to $\hat{\omega}(x,t)$ in the second phase corresponds to an extension of the fringe pattern established in the first phase, with increasing distance of the detection plate.
The simulations $\rho(x,t)$ of the SE model exhibit the space-time dilation property, with local minima remaining on the $x$-axis. The simulations 
$\omega(x,t)$ of the NLAD model exhibit an elevation of local minima above the $x$-axis and a dissipation of the fringe pattern peaks.

\begin{figure}[htp]
\begin{center}
\includegraphics[width=4.8in,height=2.3in]{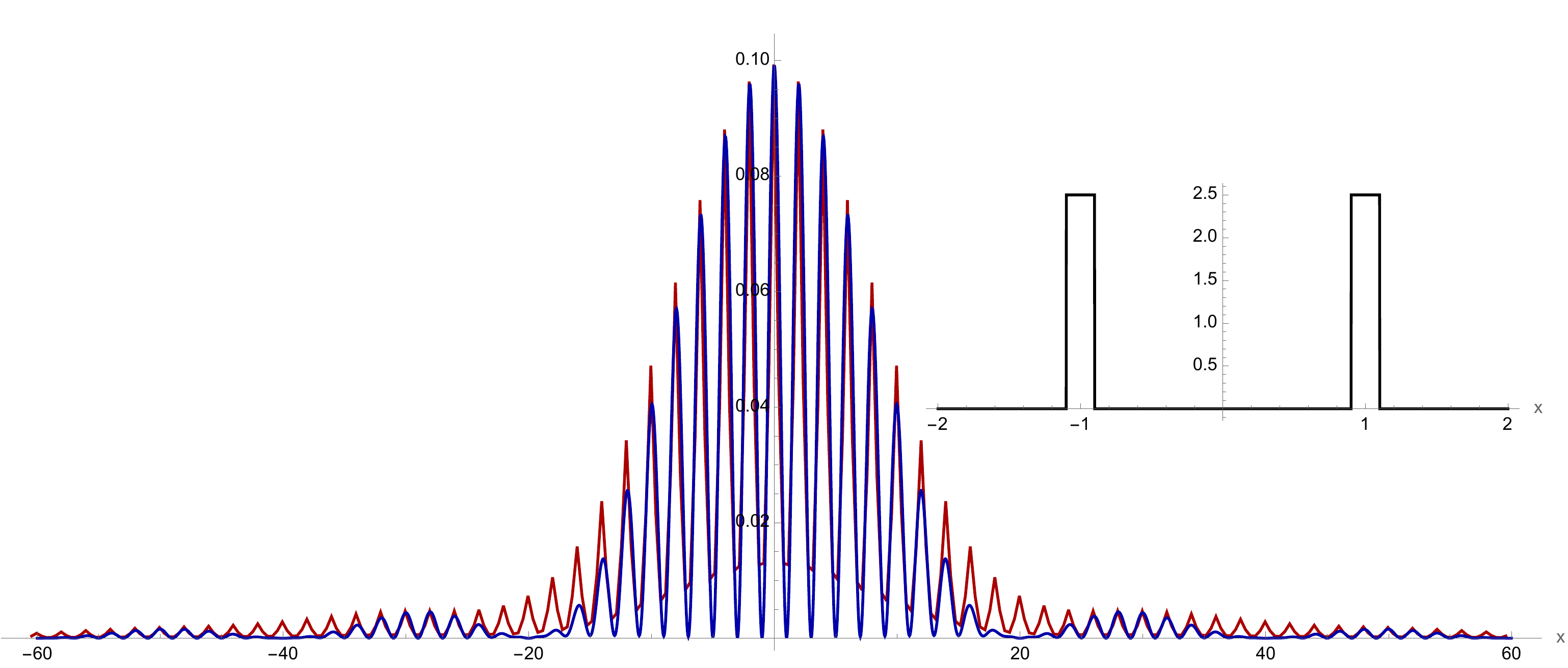}\\
\caption{The interference diffraction pattern of the solution $\rho(x,t_2)$ of SE  (blue) in equation (\ref{Eq2.3}) and the dilated solution $\hat{\omega}(x,t_2)$ 
(red) in the NLAD equation (\ref{Eq3.4}).
$\alpha = 1/\pi^3$, $s = 1$, $b =.1$, $\beta_0 = 1 / (8 \, b^2)$, $\beta_1 = \pi / (2 \, b \, 15^2)$, $\beta_2 =  \pi / (2 \, b \, 25^2)$, 
$t_2 = 2.0 / \pi$.}
\label{Fig8}
\end{center}
\end{figure}

\begin{figure}[htp]
\begin{center}
\includegraphics[width=4.8in,height=2.3in]{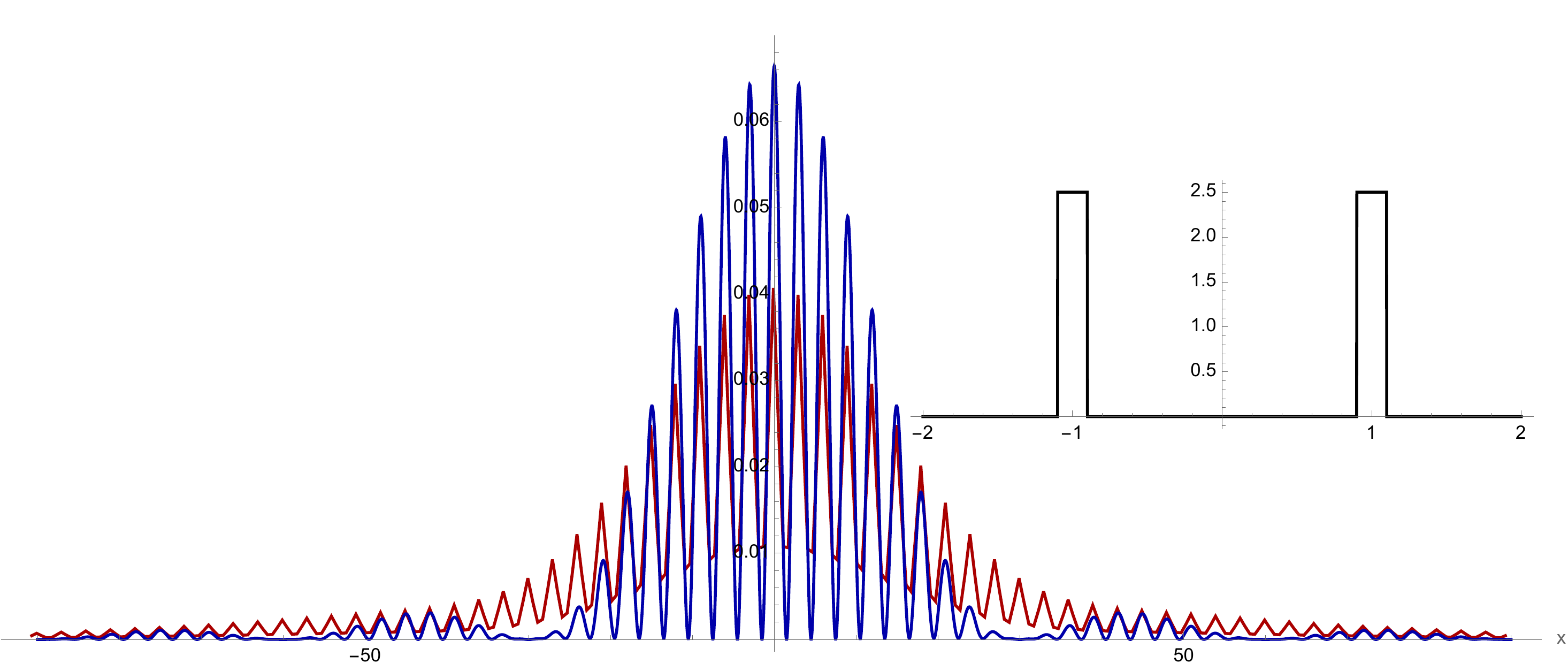}\\
\caption{The interference diffraction pattern of the solution $\rho(x,t_3)$ of SE  (blue) in equation (\ref{Eq2.3}) and the dilated solution $\hat{\omega}(x,t_3)$ 
(red) in the NLAD equation (\ref{Eq3.4}).
$\alpha = 1/\pi^3$, $s = 1$, $b =.1$, $\beta_0 = 1 / (8 \, b^2)$, $\beta_1 = \pi / (2 \, b \, 15^2)$, $\beta_2 =  \pi / (2 \, b \, 25^2)$, 
$t_3 = 3.0 / \pi$.}
\label{Fig9}
\end{center}
\end{figure}

\begin{figure}[htp]
\begin{center}
\includegraphics[width=4.8in,height=2.3in]{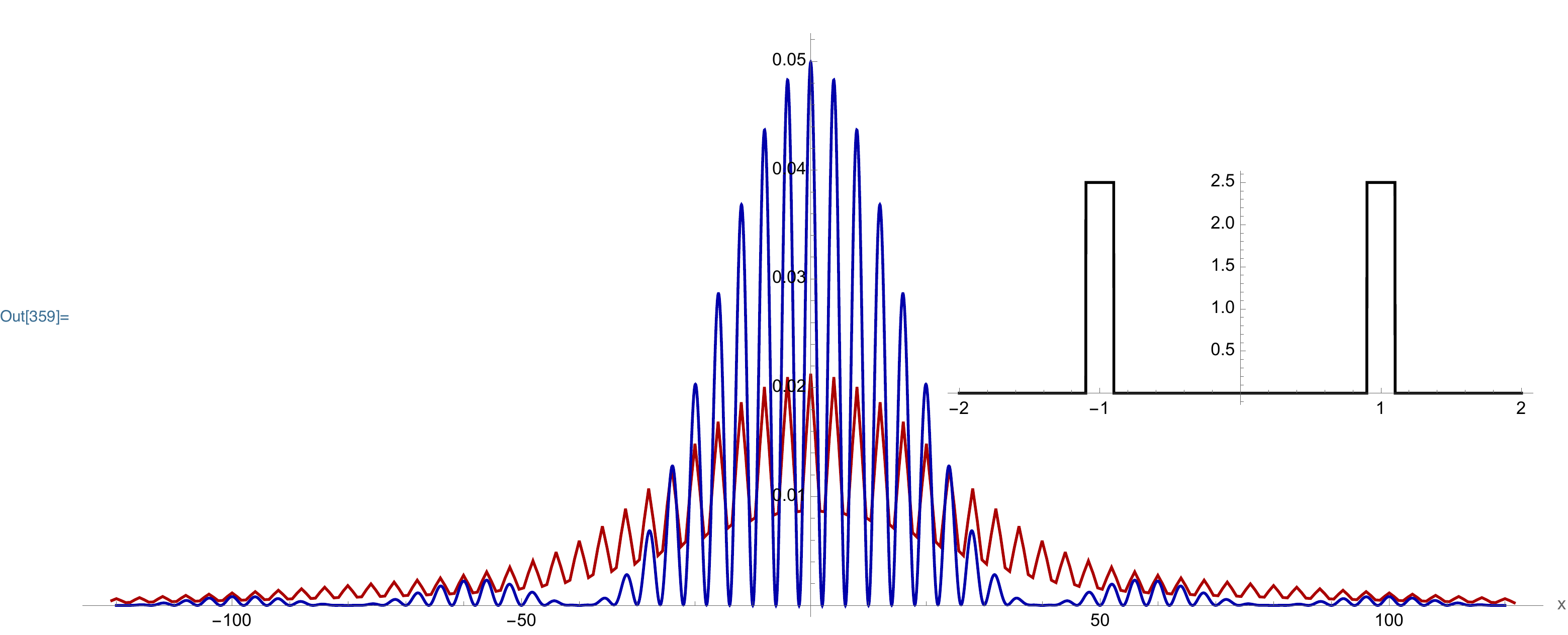}\\
\caption{The interference diffraction pattern of the solution $\rho(x,t_4)$ of SE  (blue) in equation (\ref{Eq2.3}) and the dilated solution $\hat{\omega}(x,t_4)$ 
(red) in the NLAD equation (\ref{Eq3.4}).
$\alpha = 1/\pi^3$, $s = 1$, $b =.1$, $\beta_0 = 1 / (8 \, b^2)$, $\beta_1 = \pi / (2 \, b \, 15^2)$, $\beta_2 =  \pi / (2 \, b \, 25^2)$, 
$t_4 = 4.0 / \pi$.}
\label{Fig10}
\end{center}
\end{figure}

\begin{figure}[htp]
\begin{center}
\includegraphics[width=4.8in,height=2.3in]{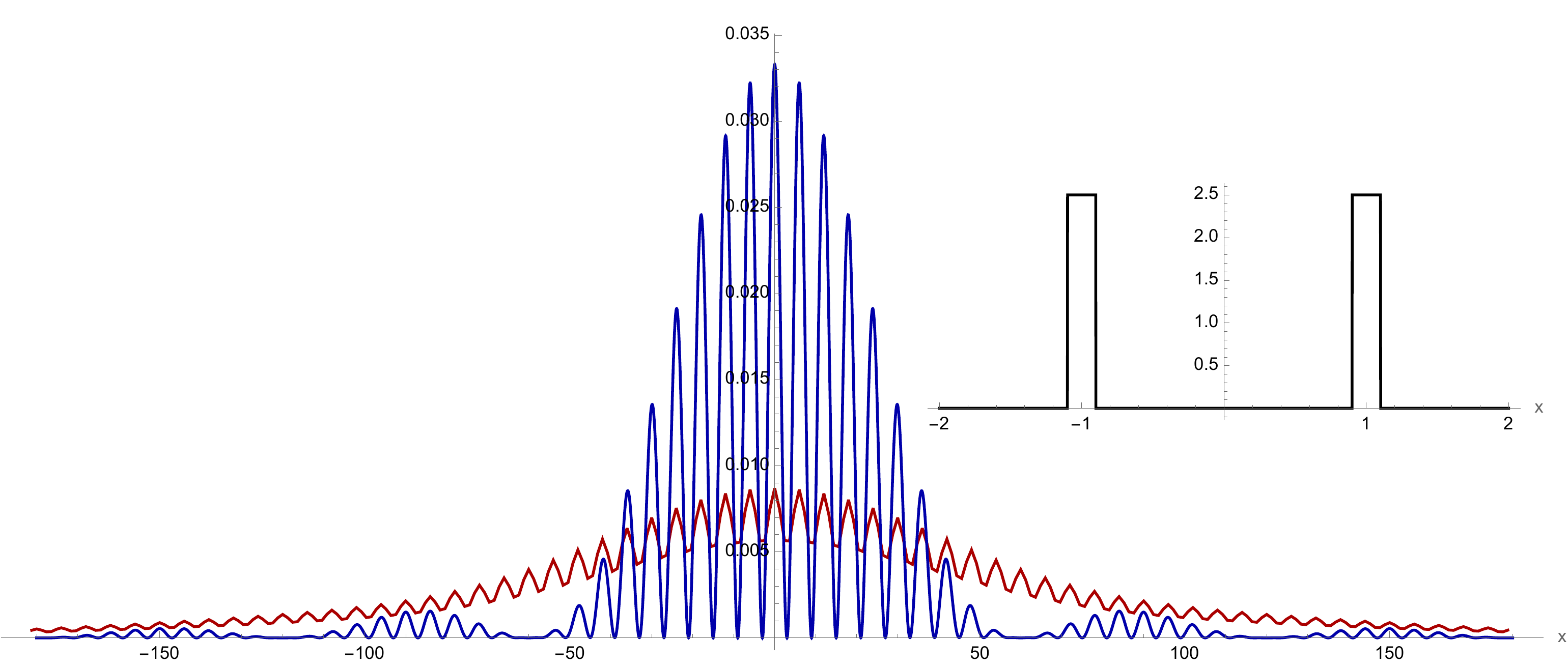}\\
\caption{The interference diffraction pattern of the solution $\rho(x,t_6)$ of SE  (blue) in equation (\ref{Eq2.3}) and the dilated solution $\hat{\omega}(x,t_6)$ 
(red) in the NLAD equation (\ref{Eq3.4}).
$\alpha = 1/\pi^3$, $s = 1$, $b =.1$, $\beta_0 = 1 / (8 \, b^2)$, $\beta_1 = \pi / (2 \, b \, 15^2)$, $\beta_2 =  \pi / (2 \, b \, 25^2)$, 
$t_6 = 6.0 / \pi$.}
\label{Fig11}
\end{center}
\end{figure}

\newpage
\section{Summary}
\label{s:summary}

We have developed a nonlocal advection-diffusion model NLAD for the two-slit experiment of quantum mechanics, in which quantum particles are projected forward, one at a time, through two stits, and detected downstream on a detection surface. Our work here is an extension of the models in \cite{Webb2} and \cite{Webb3}, which allowed higher order fringe pattern levels observed in experiments (Figure \ref{Fig2}). We compare the NLAD  equation model to the Schr\"odinger equation model SE with initial data consisting of two rectangular steps. 

In both formulations NLAD and SE, there is a two-phase development of the multi-level fringe pattern.  In the first phase the initial data transitions to an established multi-level fringe pattern with local minima located approximately on the $x$-axis. This transition is very simple for the NLAD model, but very complex for the SE model.

In the second phase the multi-level fringe pattern established in the first phase evolves in a space-time dilation in both models. In the second phase of the SE model, the pattern established in the first phase is preserved almost perfectly in the space-time dilation with constant speed. In the second phase of the NLAD model, the pattern established in the first phase dissipates, with local minima rising above the $x$-axis and with magnitude of the fringe pattern oscillations decreasing.

The SE formulation and the NLAD formulation of the 2-slit experiment have very different interpretations. The interpretation of the SE model is that an individual  quantum particle exists as a wave moving forward in space. The interpretation of the NLAD model is that an individual quantum particle exists as a component of an ensemble, with its forward movement influenced by nonlocal reaction to its environment within a sensing radius of its spatial position.

Scientific understanding of the two-slit experiment, which is one of the most fundamental experiments in science, requires mathematical formulations, which provide descriptive connection to experiments, and interpretation of physical processes. The two formulations SE and NLAD can be compared for these purposes.

\end{document}